\documentclass[letter,12pt]{article}
\usepackage{amsmath}
\usepackage{enumerate}
\usepackage{amssymb}
\usepackage{geometry}
\usepackage{setspace}
\usepackage{graphicx}
\usepackage{natbib}
\usepackage{lmodern}
\usepackage{soul}
\usepackage{comment}
\usepackage{ifthen}
\usepackage{xr}
\usepackage{framed}
\newtheorem{theorem}{Theorem}

\newtheorem{definition}{Definition}
\newtheorem{lemma}{Lemma}
\newtheorem{proposition}{Proposition}

\newenvironment{proof}[1][Proof]{\textbf{#1.} }{\ \rule{0.5em}{0.5em}}
\DeclareMathSizes{12}{11}{7.7}{5.5}
\usepackage[inline]{enumitem}
\usepackage{comment}

\graphicspath{ {WD-Graphs/} }

\begin{document}

\title{\textbf{ Failure of Smooth Pasting Principle and Nonexistence of Equilibrium Stopping Rules under Time-Inconsistency}
}

\thispagestyle{empty}
\title{Failure of Smooth Pasting Principle and Nonexistence of Equilibrium Stopping Rules under Time Inconsistency\thanks{Zhou gratefully acknowledges financial supports through a start-up grant at Columbia University and through the Nie Center for Intelligent Asset Management. The authors thank the associate editor and two anonymous referees for their detailed and constructive comments that have led to a much improved version.}}
\author{Ken Seng Tan\and Wei Wei\and Xun Yu Zhou %
\thanks{Tan: Department of Statistics and Actuarial Science
University of Waterloo
Mathematics 3, 200 University Avenue West
Waterloo, Ontario, Canada N2L 3G1. E-mail: kstan@uwaterloo.ca.  %
Wei: Department of Statistics and Actuarial Science, University of Waterloo, 200 University Avenue West, Waterloo, N2L 3G1, ON, Canada. E-mail: wei.wei@uwaterloo.ca. 
Zhou: Department of IEOR, Columbia University, New York, NY 10027, USA. Email: xz2574@columbia.edu.
}
}
%
%
\thispagestyle{empty}
\date{\today}
\maketitle
\thispagestyle{empty}

\vspace*{-0.4cm}
\begin{abstract}
This paper considers time-inconsistent stopping problems in which the inconsistency arises from a class of non-exponential discount functions called  the weighted discount functions.  We show that 
the smooth pasting principle, the main approach  that is used to construct explicit solutions for the classical  time-consistent optimal stopping problems,
may fail under time-inconsistency. Specifically, we prove that the smooth pasting solves a time-inconsistent problem, within the intra-personal game theoretic framework with a general nonlinear cost functional and a geometric Brownian motion,  if and only if certain inequalities on the model primitives are satisfied. In the special case of a real option problem, we show that the violation of these inequalities can happen even for very simple non-exponential discount functions. Moreover, we show that the real option problem actually does not admit any equilibrium whenever the smooth pasting approach fails. The negative results in this paper 
caution blindly extending the classical approach for time-consistent
stopping problems to their time-inconsistent counterparts.

\end{abstract}


\newpage
\setcounter{page}{1}
\section{Introduction}\label{SecIntroduction}
A crucial assumption imposed on classical optimal stopping models is that an agent has a constant time preference rate and hence discounts her future payoff exponentially. When this assumption is violated, an optimal stopping problem becomes generally time-inconsistent in that any optimal stopping rule obtained today may no longer be optimal from the perspective of a  future date.
The problem then becomes largely descriptive rather than normative because there is generally no {\it dynamically} optimal solution that can be used to guide the agent's decisions. Different agents may react differently to a same time-inconsistent problem, and a goal of the study is to {\it describe} the different behaviors.
\cite{str1955} is the first to observe that non-constant time preference rates result in time-inconsistency, and to categorize three types of agents when facing such time-inconsistency. One of the types is called a ``non-committed, sophisticated agent" who, at {\it any} given time, optimizes the underlying objective taking as constraints the stopping decisions chosen by her future selves. Such a problem
has been  formulated within an intra-personal game theoretic framework and the corresponding equilibria are used to describe the behaviors of this type of agents;  see, for example, \cite{phepol1968}; \cite{lai1997}; \cite{odorab1999}; \cite{krusmi2003} and \cite{lutmar2003}. An extended dynamic programming equation
for continuous-time deterministic equilibrium controls is derived in \cite{ekelaz2006}, followed by a stochastic version in \cite{bjork2010general} and application to a mean--variance portfolio model in \cite{bjomurzho2014}.

\par

This paper studies a time-inconsistent stopping problem in continuous time within the intra-personal game framework, in which the source of time-inconsistency is the so-called weighted discount function (WDF), a very general class of non-exponential discount functions.\footnote{The WDF, proposed in \cite{ebert2016weighted},  is a weighted average of a set of exponential discount functions. It has been shown in \cite{ebert2016weighted} that it can be used to model the time preference of a group of individuals  as well as that of behavioral agents, and that most commonly used non-exponential discount functions
are WDFs.} We  make two main contributions. First, we demonstrate that the smooth pasting (SP) principle, which is almost the {\it exclusive} approach in solving classical optimal stopping problems, may fail when time-consistency is lost.
Second,  for a stopping model whose time-consistent counterpart is the
well-studied real option problem, we establish a condition under which no equilibrium
stopping rule exists.
 These results are constructive and they caution blindly extending the SP principle to time-inconsistent stopping problems. 
\par
Let us now elaborate on the first contribution. Recall that the SP is used to  derive (often explicit) solutions to conventional, time-consistent optimal stopping problems.
It conjectures a candidate solution to
the underlying Bellman equation (or variational inequalities), which is
a free boundary PDE,  based on the $C^1$ smooth pasting around the free boundary, and then checks that it solves the PDE under some standard regularity/convexity conditions on the model
primitives. Finally it verifies that the first hitting time of the free boundary
indeed solves the optimal stopping problem using the standard verification technique. Recently, \cite{grewan2007} and \cite{hsiaw2013goal}, among others, extend the application of the SP principle to solving time-inconsistent stopping problems. While the SP {\it happened to} work
in the specific settings of these papers, it is more an exception than a rule.
Indeed, in the present  paper we show that, for a geometric Brownian motion with a nonlinear
cost functional,  while the SP  always yields a candidate solution, the latter actually gives rise to an
equilibrium stopping rule if and only if certainty inequalities on the model primitives are satisfied.
These inequalities hold trivially for the time-consistent exponential discount case, but  does not in general for its time-inconsistent non-exponential counterpart, {\it even if all the other parameters and assumptions (state dynamics, running cost, etc.) are identical}. Indeed, the violation of such inequalities  is not rare even in very simple cases.
For example, we show that in the  special case of a  real option problem with some WDFs including the pseudo-exponential discount function  (\citeauthor{ekelaz2006} \citeyear{ekelaz2006}; \citeauthor{kar2007} \citeyear{kar2007}; \citeauthor{harlai2013} \citeyear{harlai2013}), the inequalities do not hold for plausible sets of parameter values of the chosen discount functions.
%
%
%
The bottom line is that
one cannot blindly apply the  SP to any stopping model when time-inconsistency is present, even if the SP does work for its time-consistent counterpart. \par

The second contribution is on the nonexistence of an intra-personal equilibrium. For a time-consistent stopping problem, optimal stopping rules exist when the cost functional and the underlying process satisfy some mild regularity conditions (see, e.g., \citeauthor{peskir2006optimal} \citeyear{peskir2006optimal}). However, this is no longer the case for the time-inconsistent counterpart. To demonstrate this, we again take the   real option problem with a WDF. For such a problem, we prove that there simply does not exist any equilibrium stopping rule whenever the
aforementioned inequalities are violated and hence the SP principle fails. Our result therefore reveals that equilibrium stopping rules within the intra-personal game theoretic framework may not exist no matter what regularity conditions are imposed on the underlying models. \par

There are studies in the literature on time-inconsistent stopping including  nonexistence results, albeit in considerably different settings especially in terms of the source of  time-inconsistency  and the definition of an equilibrium.
\cite{bayraktar2018time}
consider a stopping problem with a discrete-time Markov chain, whereas  the time-inconsistency comes from the mean-variance objective functional.
The Markov chain takes value in a set of finite numbers, which allows them to discuss the nonexistence of equilibrium stopping rules by enumeration.
\cite{christensen2018time} and \cite{christensen2018finding}
study continuous-time  stopping problems where the time-inconsistency follows from the types of payoff functions (mean-variance or  endogenous habit formation). In particular, \cite{christensen2018time} show that the candidate solution derived from the SP may not lead to an equilibrium stopping for some range of parameters. However, the definition of equilibria  in these papers is entirely different from the one
based on the ``first-order" spike variation; the latter seems to be widely adopted by many papers (see, e.g.,  \citeauthor{odorab1999}  \citeyear{odorab1999};  \citeauthor{bjork2010general} \citeyear{bjork2010general}; \citeauthor{ekeland2012time} \citeyear{ekeland2012time}; and \citeauthor{bjomurzho2014} \citeyear{bjomurzho2014}) including the present one.\footnote{\cite{christensen2018time} also consider mixed strategies as opposed to the pure strategies studied in our paper and many other papers. Time-inconsistent problems using mixed strategies are interesting, and it is possible that no equilibrium may be found in the class of mixed strategies either. However,
the main point of this paper is to show that a change of discounting factor from exponential to
non-exponential  may cause a stopping problem that has an equilibrium to one that does not, {\it even though both are using pure strategies}.}
\cite{huang2018time} investigate a continuous-time stopping problem with non-exponential discount functions. They define an equilibrium via a fixed point of
a mapping, which is essentially based on a zeroth-order condition and hence  is different from our definition. Under their setting, immediately  stopping is always a (trivial) equilibrium (so there is no issue of nonexistence), which is not the case according to  our definition.

The remainder  of the paper is organized as follows.
In Section \ref{GeneralModel}, we  recall the definition and some important properties of the
WDF introduced by  \cite{ebert2016weighted}, formulate a general time-inconsistent stopping problem within the intra-personal game theoretic framework, and characterize the equilibrium stopping rules by a Bellman system and provide the verification theorem. In Section \ref{Sec:SPGeneral} we consider the case when the state process is a geometric Brownian motion, apply the SP principle to derive a candidate solution  and establish certain equivalent conditions for the derived candidate solution to actually solve the Bellman system. Then we present a real option problem, in which the aforementioned equivalent conditions reduce to a single inequality, failing which there is simply no equilibrium at all. Finally, Section \ref{Conclusion} concludes the paper. Appendix A contains proofs of some results.

\section{The Model }\label{GeneralModel}
\subsection{Time preferences}

Throughout this paper we consider weighted discount functions  defined as follows.

\begin{definition}[\citeauthor{ebert2016weighted} \citeyear{ebert2016weighted}]\label{WeightedDiscounFunction}
Let $h: [0,\infty) \to (0,1]$ be strictly decreasing with $h(0)=1.$ We call $h$ a weighted discount function (WDF) if there exists a distribution function $F$ concentrated on $[0,\infty)$ such that
\begin{align}\label{weightedFormula}
h(t) = \int^{\infty}_0e^{-rt}dF(r).
\end{align}
Moreover, we call $F$  the \textnormal{weighting distribution} of $h.$
\end{definition}
Many commonly used discount functions can be represented in weighted form. For example, exponential function $h(t) =e^{-rt},r>0$ (\citeauthor{sam1937} \citeyear{sam1937}) and pseudo-exponential function $h(t) = \delta e^{-rt}+(1-\delta)e^{-(r+\lambda)t},0<\delta <1, r>0,\lambda>0$ (\citeauthor{ekelaz2006} \citeyear{ekelaz2006}; \citeauthor{kar2007} \citeyear{kar2007}) are WDFs with degenerate and binary distributions respectively. A more complicated example is  the generalized hyperbolic discount function (\citeauthor{loewenstein1992anomalies} \citeyear{loewenstein1992anomalies}) with parameters $\gamma > 0, \beta > 0$, which  can be represented as
\begin{align}\label{ghd}
h(t)=\frac{1}{(1+\gamma t)^{\frac{\beta}{\gamma}}}\equiv \int^{\infty}_0e^{-rt}f\left(r;\frac{\beta}{\gamma},\gamma \right)dr
\end{align}
where $
f(r;k,\theta)=\frac{r^{k-1}e^{-\frac{r}{\theta}}}{\theta^k\Gamma(k)}
$
is the density function of the Gamma distribution with parameters $k$ and $\theta$,   and $\Gamma(k)=\int^{\infty}_0x^{k-1}e^{-x}dx$  the Gamma function evaluated at $k$. See \cite{ebert2016weighted} for more examples and
discussions about the types of discount functions that are of weighted form. \par
The following result is a restatement of the well-known Bernstein's theorem in terms of WDFs, which actually provides a characterization of the latter. 
\begin{theorem}[\citeauthor{ber1928} \citeyear{ber1928}]\label{the:Bernstein}
 A discount function $h$ is a WDF if and only if it is continuous on $[0,\infty),$ infinitely differentiable on $(0,\infty),$ and satisfies $(-1)^{n}h^{(n)}(t)\ge 0,$ for all non-negative integers $n$ and for all $t>0.$
\end{theorem}
Bernstein's theorem can be used to examine if a given function is a WDF without necessarily representing it in the form of (\ref{weightedFormula}). For example,
it follows from this theorem that the constant sensitivity discount function $h(t) = e^{-at^k}, a,k>0,$ and the constant absolute decreasing impatience   discount function $h(t)=e^{e^{-ct}-1},c>0$, are both WDFs.

\subsection{Stopping rules and equilibria}
On a complete filtered probability space $(\Omega,\mathcal{F},\mathbb{F}=\{\mathcal{F}_t\}_{t\ge 0},P)$ there lives a one dimensional Brownian motion $W$, and
a family of Markov diffusion processes $X=X^x$ parameterized by the initial state $X_0=x\in\mathbb{R}$ and governed by the following stochastic differential equation (SDE)
\begin{align}\label{SDE}
dX_t = b(X_t)dt+\sigma(X_t)W_t,\;\;X_0=x,
\end{align}
where $b, \sigma$ are Lipschitz continuous functions, i.e., there exists an $L>0$ such that for any $x\neq y$
\begin{align}\label{Lip}
\vert b(x) - b(y)\vert + \vert\sigma(x)-\sigma(y)\vert\le L\vert x - y\vert.
\end{align}
We assume that $\mathbb{F}$  is the $P$-augmentation of the natural filtration generated by $X$. To avoid an uninteresting case we also assume that $|\sigma(x)|\geq c>0$ $\forall x\in\mathbb{R}$ so that $X$ is non-degenerate.\footnote{Here we assume that the Brownian motion is one dimensional just for notational simplicity. There is no essential difficulty with a multi-dimensional Brownian motion.}

For any fixed $x\in\mathbb{R}$, an agent monitors the process $X=X^x$ and aims to minimize the following cost functional
\begin{align}\label{StationaryCostFunctionalGeneral}
J(x;\tau) = \mathbb{E}\left[\int^{\tau}_0h(s)f(X_s)ds+h(\tau)g(X_\tau)\Big\vert X_0=x\right]
\end{align}
by choosing $\tau\in \mathcal{T}$, the set of all  $\mathbb{F}$-stopping times.  Here
$h$ is a WDF with a weighting distribution $F$, $g$ is continuous and bounded, and $f$ is continuous  with polynomial growth, i.e., there exists $m\ge 1$ and $ C>0$ such that
\begin{align}\label{GrowthPolynomial}
\vert f(x)\vert\le C(\vert x\vert^m+1).
\end{align}
Moreover, we assume  that there exists $n\ge 1, C(r)>0$ satisfying
$\int^{\infty} _0C(r)dF(r)+\int^{\infty}_0rC(r)dF(r)<\infty$ such that
\begin{align}\label{GrowthPolynoimalCostFunctional}
\sup_{\tau\in\mathcal{T}}\mathbb{E}\left[\int^{\tau}_0e^{-rs}\vert f(X_s) \vert ds+e^{-r\tau}\vert g(X_{\tau})\vert \Big\vert X_0=x\right]\le C(r)(\left\vert x \right\vert^n+1 ),
\;\;\forall r\in \text{supp}(F).
\end{align}
This is a (weak) assumption to ensure that the optimal value of the stopping problem is finite, and hence the problem is well-posed.

We now define stopping rules which are essentially binary feedback controls. These stopping rules induce Markovian stopping times for any given Markov process.

\begin{definition}[Stopping rule]\label{StoppingRuleNon}
A {stopping rule} is a measurable function $u:$ $\mathbb{R}\to \{0,1\}$ where $0$ indicates ``continue'' and $1$ indicates ``stop''. For any given Markov
process $X=\{X_t\}_{t\ge 0}$, a stopping rule $u$ defines a Markovian  stopping time
\begin{equation}\label{StoppingTimeNon}
\tau_u=\inf\{t\geq 0,u(X_t)=1\}.
\end{equation}
%
\end{definition}

Given a stopping rule $u$, we can define the stopping region
$\mathcal{S}_u=\{x\in(0,\infty): u(x)=1\}$.
For any $x\in\bar{\mathcal{S}_u}$,
since the underlying process $X$ is non-degenerate, a standard result (e.g., Chapter $3$ of \citeauthor{itoandh1965diffusion} \citeyear{itoandh1965diffusion}) yields that $\mathbb{P}(\tau_u=0\vert X_0=x)=1,$ and hence $J(x;\tau_u) = g(x).$ This means that
the agent stops immediately once the process reaches at any point  in
$\bar{\mathcal{S}_u}$. As a result, in the setting of this paper,
the  continuation region is
$\mathcal{C}_u = \bar{\mathcal{S}_u}^c.$


As discussed earlier the  non-exponential discount function $h$ in the cost functional
(\ref{StationaryCostFunctionalGeneral}) renders the underlying optimal stopping problem
generally time-inconsistent. 
In this paper we consider a sophisticated and non-committed  agent who is aware of the time-inconsistency but unable to control her future actions.
In this case, she seeks to find the so-called equilibrium strategies within the intra-personal game theoretic framework, in which the individual is represented by different players at different dates.\footnote{Given the infiniteness of the time horizon, the stationarity of the process $X$ as well the time-homogeneity of the running objective function $f$, each self at any given time $t$ faces exactly the same decision problem as the others, which only depends on the current state
$X_t = x$, but not on time $t$ directly. We can thus identify self $t$ by the current state $X_t = x$. That is why we need to consider only
stationary stopping rules $u$, which are functions of the state variable $x$ only. For details on this convention, see, e.g. \citeauthor{grewan2007} (\citeyear{grewan2007}); \citeauthor{ekeland2012time} (\citeyear{ekeland2012time}); \citeauthor{harlai2013} (\citeyear{harlai2013}) and, in particular, Section 3.2 of \citeauthor{ebert2016weighted} (\citeyear{ebert2016weighted}).}\par

We now give the precise definition of an equilibrium stopping rule $\hat{u}$, which essentially entails a solution to a game in which
 no self at any time (or, equivalently in the current setting, at any state) is willing to deviate from $\hat{u}$.
\begin{definition}[Equilibrium stopping rule]\label{Def:EquilibriumStopping}
The stopping rule $\hat{u}$ is an equilibrium stopping rule if
\begin{equation}\label{StationaryDefinitionInequalityNon}
\limsup_{\epsilon\rightarrow 0+}\frac{J(x;\tau_{\hat{u}})-J(x;\tau^{\epsilon,a})}{\epsilon}\le 0,\;\;\forall x\in \mathbb{R},\; \forall  a\in\{0,1\},
\end{equation}
where
\begin{equation}\label{uAuxiliaryStationaryNon}
\tau^{\epsilon,a} =\left\{
  \begin{array}{ccr}
    \inf\{t\geq\epsilon ,\hat u(X_t)=1\} & \quad \text{if $a=0$},\\
    0 & \quad \text{if $a=1$}
  \end{array}
\right.
\end{equation}
with $\{X_t\}_{t\ge 0}$ being the solution to (\ref{SDE}).
\end{definition}

This definition of an equilibrium is consistent with the majority of definition  for time-inconsistent control problems in the literature (see, e.g.,  \citeauthor{bjork2010general} \citeyear{bjork2010general}; \citeauthor{ekeland2012time} \citeyear{ekeland2012time}; and \citeauthor{bjomurzho2014} \citeyear{bjomurzho2014}) when  a stopping rule is interpreted as a binary control. Indeed, $\tau^{\epsilon,a}$ is a stopping time that might be different from $\tau_{\hat u}$
only in the very small initial time interval $[0,\epsilon)$; hence it is a ``pertubation" of the latter.
%
%
\subsection{Equilibrium characterization}
The following result,
Theorem \ref{EquilibriumcharacterizationNon}, formally establishes the Bellman system and provides the verification theorem for verifying equilibrium stoppings.
\begin{theorem}[Equilibrium characterization]\label{EquilibriumcharacterizationNon}
Consider the cost functional (\ref{Costfunctional}) with WDF $h(t)=\displaystyle\int^{\infty}_0e^{-rt}dF(r),$ a stopping rule $\hat{u},$ an underlying process $X$ defined by (\ref{SDE}), functions $w(x;r)=\mathbb{E}\left[\displaystyle\int^{\tau_{\hat{u}}}_0e^{-rt}f(X_t)dt+e^{-r\tau_{\hat{u}}}g(X_{\tau_{\hat{u}}})\Big\vert X_0=x\right]$ and $V(x)=\displaystyle\int_0^{\infty}w(x;r)dF(r).$ %
Suppose that $w$ is continuous in $x$ and $V$ is continuously differentiable with its first-order derivative being absolutely continuous. If  $(V,w,\hat{u})$ solves
\begin{align}
&\min\left\{\frac{1}{2}\sigma^2(x)V_{xx}(x)+b(x)V_x(x)+f(x)-\int_0^{\infty}rw(x;r)dF(r),
g(x)-V(x)\right\}=0,\;x\in \mathbb{R}, \label{StationaryPDENon} \\
&\hat{u}(x) =\left\{
  \begin{array}{rcr}
    1 & \quad \text{if $V(x)=g(x)$},\\
    0 & \quad \text{otherwise},
  \end{array}
\right.\; x\in \mathbb{R},\label{StationaryuhatContinuousNon}
\end{align}
then $\hat{u}$ is an equilibrium stopping rule and $V$ is the value function of the problem, i.e., $V(x)=J(x;\tau_{\hat{u}})\;\forall x\in \mathbb{R}.$
\end{theorem}
A proof to the above proposition is relegated to the appendix.\footnote{A proof of this result in a different setting was provided in \cite{ebert2016weighted}. Here we supply a proof for reader's convenience. }

\section{Failure of SP  and Nonexistence of Equilibrium}\label{Sec:SPGeneral}
In the classical literature on (time-consistent) stopping, optimal solutions are often obtained by the SP, because the candidate solution obtained from the SP  must solve the Bellman system (and hence the optimal stopping problem) under some mild conditions, such as the smoothness and convexity/concavity of the cost functions. 
In economics terms, the SP principle
amounts to the matching of the marginal cost at the stopped state; hence some economists apply the SP principle without even
explicitly introducing the Bellman system.
However, as we will show  in this section,  the SP approach in the presence of time-inconsistency may {\it not} yield a solution to the Bellman system (and therefore {\it not} to the stopping problem within the game theoretic framework), no matter how smooth and convex/concave the cost functions might be.\par

\subsection{A time-consistent benchmark}\label{tcc}

Let us start with a time-consistent optimal stopping problem which we use as a benchmark for comparison purpose and outline the way to use the SP principle in constructing explicit solutions. Consider the following classical optimal stopping problem
\begin{align}\label{OptimalStoppingConvention}
\inf_{\tau\in\mathcal{T}} \mathbb{E}\left[\int^{\tau}_0e^{-rs}f(X_s)ds+e^{-r\tau}K\Big\vert X_0=x \right],
\end{align}
where  the underlying process $X$ is a geometric Brownian motion
\begin{align}
dX_t = bX_tdt + \sigma X_tdW_t,\;x>0,
\end{align}
and
$\mathcal{T}$ is the set of all stopping times with respect to $\mathbb{F}$.\footnote{In this formulation the final cost is assumed to be
a constant lump sum $K$ without loss of generality. In fact, by properly modifying the running cost, we are able to reduce the stopping problem with a final  cost function $g$ to one with a  final  cost being any given constant $K>0$. To see this,
applying Ito's formula to $e^{-rt}(g(X_t)-K)$, we get
\begin{align*}
&\mathbb{E}\left[\int^{\tau}_0e^{-rs}f(X_s)ds+e^{-r\tau}g(X_{\tau})\Big\vert X_0=x \right]\\=&\mathbb{E}\left[\int^{\tau}_0e^{-rs}f(X_s)ds+e^{-r\tau}K+e^{-r\tau}(g(X_{\tau})
-K)\Big\vert X_0=x \right]\\=&\mathbb{E}\left\{\int^{\tau}_0e^{-rs}f(X_s)ds+e^{-r\tau}K
+\int^{\tau}_0e^{-rs}\left[\frac{1}{2}\sigma^2x^2g_{xx}(X_s)+bxg_x(X_s)-r(g(X_s)-K)\right]ds\Big\vert X_0=x \right\}.
\end{align*}
Letting $\tilde{f}(x):=f(x)+\frac{1}{2}\sigma^2x^2g_{xx}(x)+bxg_x(x)-r(g(x)-K)$, the cost functional now becomes the one in problem (\ref{OptimalStoppingConvention}) with running cost $\tilde{f}.$}

In what follows we assume that the running cost $f$ is continuously differentiable, increasing and concave. Moreover, to rule out the ``trivial cases'' where
either immediately stopping or never stopping is optimal
for this time consistent benchmark, we assume that $f(0)<rK, b<r,$ and $\lim_{x\to \infty}f_x(x)x = \infty$.\par
Define $L(x;r) = \mathbb{E}[\int^{\infty}_0e^{-rs}f(X_s)ds\vert X_0=x].$  Noting that $X$ is a geometric Brownian motion, we have after straightforward manipulations
\begin{align}\label{SpecialSolution}
L(x;r) = \int^{\infty}_0\int^{\infty}_0f(yx)e^{-rs}G(y,s)dyds,
\end{align}
where $G(y,s) = \frac{1}{\sqrt{2\pi}}\frac{1}{\sigma y\sqrt{s}}e^{-\frac{(\ln y-(b-\frac{1}{2}\sigma^2)s)^2}{2\sigma^2 s}}$.
To ensure $L$ and $L_x$ are well defined,  we further assume that $f$ has linear growth and $f_x(0+)<\infty$. \par
We now characterize the optimal stopping rule as follows.
\begin{proposition}\label{Prop:Bechmark}
There exists  $x_B>0$ such that the stopping rule $u_B(x)={\bf 1}_{x\ge x_B}(x)$
solves optimal stopping problem (\ref{OptimalStoppingConvention}). Moreover, $x_B$ is the unique solution of the following algebraic equation in $y$:
\begin{align}\label{SPEquation}
\alpha(r)[K-L(y;r)] + L_x(y;r)y = 0
\end{align}
where \begin{align}\label{Root}
\alpha(r) = \frac{-(b-\frac{1}{2}\sigma^2)+\sqrt{(b-\frac{1}{2}\sigma^2)^2+2\sigma^2r}}{\sigma^2}.
\end{align}
\end{proposition}
The key to proving this theorem is to make us of the SP; see Appendix A.2.

\subsection{Equivalent conditions under time-inconsistency}

We now consider exactly the same stopping problem as the above time-consistent benchmark except that the exponential discount function is replaced by  a WDF, namely, the cost  functional is changed to
\begin{align}\label{SPCostFunctional}
J(x;\tau) = \mathbb{E}\left[\int^{\tau}_0h(s)f(X_s)ds+h(\tau)K\Big\vert X_0=x\right],
\end{align}
where $h$ is a WDF with a weighting distribution $F$.\par
As in the case of exponential discounting, we need to impose the following regularity  conditions on the parameters of the problem:
\[ b<r,\; \forall r\in\text{supp}(F); \;\mbox{and } \max\left\{\int^{\infty}_0\frac{1}{r-b}dF(r),\int^{\infty}_0\frac{1}{r}dF(r),
\int^{\infty}_0rdF(r)\right\}<\infty.
\] These conditions either hold automatically or reduce to the respective counterparts when the discount function degenerates into the exponential one. On the other hand, they hold valid
with many genuine WDFs, including the  generalized hyperbolic discount function (\ref{ghd}) when $\gamma<\beta$ and the pseudo-exponential discount function.\par
We now attempt to use the SP principle to solve the Bellman system in Theorem \ref{EquilibriumcharacterizationNon} with the cost functional (\ref{SPCostFunctional}). We start by conjecturing that the equilibrium stopping region is $[x_*,\infty)$ for some $x_*>0$. (As in the time-consistent case, $x_*$ is called
the triggering boundary or the stopping threshold.)

It follows from the Feynman--Kac formula  that $w$ in the Bellman system is given by
\[
w(x;r) =\left\{\begin{array}{ll}
 (K-L(x_*;r))\left(\frac{x}{x_*}\right)^{\alpha(r)} + L(x;r),& x<x_*,\\
K, & x\ge x_*,
\end{array}\right.
\]
where $L(x;r)$ is defined by (\ref{SpecialSolution}) and $\alpha(r)$  by (\ref{Root}). Recall we have defined $V$ and $\hat u$ by
\[
V(x) = \int^{\infty}_0w(x;r)dF(r) =\left\{\begin{array}{ll} \int^{\infty}_0((K-L(x_*;r))(\frac{x}{x_*})^{\alpha(r)} + L(x;r))dF(r),& x<x_*,  \\
K, &x\ge x_*
\end{array}\right.
\]
and
\begin{align*}
\hat{u}(x) =\left\{
  \begin{array}{rcr}
    0 & \quad \text{$x< x_*$},\\
    1 & \quad \text{otherwise}.
  \end{array}
\right.
\end{align*}
The SP applied to $V$ ({\it not} to $w$) yields $V_x(x_*) = 0$, implying that $x_*$ is the solution to the following algebraic equation in $y$
\begin{align}\label{EquationEquilibrium}
\int^{\infty}_0\left[\alpha(r)(K-L(y;r)) + L_x(y;r)y\right]dF(r) = 0.
\end{align}
Clearly, this equation is a generalization of its time-consistent counterpart,
(\ref{SPEquation}). The following proposition stipulates that it has a unique solution.
\begin{proposition}\label{Prop:SPEquationEquilibrium}
Equation (\ref{EquationEquilibrium}) admits a unique solution in $(0,\infty).$
\end{proposition}

\begin{proof}
Following the same lines of proof of Proposition \ref{Prop:Bechmark} (Appendix A.2), we have that $Q(x):=\int^{\infty}_0(\alpha(r)(K-L(x;r)) + L_x(x;r)x)dF(r)$ is strictly decreasing in $x>0$, with $Q(0)>0$ and $Q(\infty)<0.$ This completes the proof.
\end{proof}

Proposition \ref{Prop:SPEquationEquilibrium} indicates that following  the
conventional SP line of argument does indeed give rise to a {\it candidate} solution to the Bellman system, even  under time inconsistency. We may be tempted to
claim, as taken for granted in the time-consistent case, that this candidate solution solves the Bellman system in Theorem \ref{EquilibriumcharacterizationNon} and hence the corresponding stopping rule $\hat u$  solves the equilibrium stoppping problem.  Unfortunately, this is not always the case, as shown in the following result.
\begin{theorem}\label{Prop:SPGeneral}
Assume  that $\alpha(r)[\alpha(r) - 1][K-L(x_*;r)]$ is increasing in $r\in\text{supp}(F)$, and let $x_*$ be the unique solution to (\ref{EquationEquilibrium}). Then the triplet $(V,w,\hat{u})$ solves the Bellman system in Theorem \ref{EquilibriumcharacterizationNon} and in particular
$\hat{u}$ is an equilibrium stopping rule if and only if
\begin{align}\label{InequalityGeneral2}
f(x_*) \ge \int^{\infty}_0rdF(r)K,
\end{align}
and
\begin{align}\label{InequalityGeneral1}
\int^{\infty}_0\alpha(r)[\alpha(r) - 1][K-L(x_*;r)]dF(r)+\int^{\infty}_0x_*^2L_{xx}(x_*;r)dF(r)\le 0.
\end{align}
\end{theorem}

As a proof is lengthy, we defer it to Appendix A.3.

The above theorem presents characterizing conditions (on the model primitives) for
the SP to work for stopping problems with general WDFs. These conditions are satisfied automatically in the classical time-consistent case,
but not in the time-inconsistent case in general. We will demonstrate this
with a classical real option problem in the next subsection.

\subsection{A real option problem: failure of SP and nonexistence of equilibrium}\label{Sec:example}
In this subsection we consider a special case of the model studied in the previous subsection, which is a time-inconsistent counterpart of the well-studied (time-consistent) problem
of real options. Such a problem can be used to model, among others,
when to start a new project or to abandon an ongoing project; see \cite{dixit1993art} for a systematic account
on the classical real options theory.

 %
%

The problem is to minimize
\begin{align}\label{Costfunctional}
\mathbb{E}\left[\int^{\tau}_0h(s)X_sds+h(\tau)K\Big\vert X_0=x\right]
\end{align}
by choosing $\tau\in{\cal T}$, where $X=\{X_t\}_{t\ge 0}$ is governed by\footnote{Here we assume that the geometric Brownian motion is driftless without loss of generality.}

\begin{align}\label{ExampleNon}
dX_t=\sigma X_tdW_t.
\end{align}


We now apply Theorem \ref{Prop:SPGeneral} to this problem, and see what
the equivalent conditions (\ref{InequalityGeneral2}) and (\ref{InequalityGeneral1})
boil down to.

First of all,
$$L(x;r) = \mathbb{E}[\int^{\infty}_0e^{-rt}X_tdt\Big\vert X_0=x]=\frac{x}{r}.$$
Hence

\[
w(x;r) =\left\{\begin{array}{ll}
 \left(K-\frac{x_*}{r}\right) \left(\frac{x}{x_*}\right)^{\alpha(r)}+\frac{x}{r},& x<x_*,\\
K, & x\ge x_*,
\end{array}\right.
\]
\[
V(x) =\left\{\begin{array}{ll}
\int^{\infty}_0\left(K-\frac{x_*}{r}\right)\left(\frac{x}{x_*}\right)^{\alpha(r)}dF(r)
+\int^{\infty}_0\frac{x}{r}dF(r),& x<x_*,  \\
K, &x\ge x_*,
\end{array}\right.
\]
and
\begin{align*}
\hat{u}(x) =\left\{
  \begin{array}{rcr}
    0 & \quad \text{ $x< x_*$},\\
    1 & \quad \text{otherwise},
  \end{array}
\right.
\end{align*}
where \begin{align}\label{ImportantFactorNon}
\alpha(r) = \frac{\frac{1}{2}\sigma^2+\sqrt{\frac{1}{4}\sigma^4+2\sigma^2r}}{\sigma^2}.
\end{align}
Moreover, it follows from (\ref{EquationEquilibrium}) that $x_*$ is the solution to the following equation in $y$:
\begin{align*}
\int^{\infty}_0\left(K-\frac{y}{r}\right)\alpha(r)dF(r)+\int^{\infty}_0\frac{y}{r}dF(r)=0.
\end{align*}
Thus
\begin{align}\label{TriggeringBoundaryNon}
x_*=\displaystyle\frac{\displaystyle\int^{\infty}_0\alpha(r) dF(r)}{\displaystyle\int^{\infty}_0\frac{\alpha(r)-1}{r} dF(r)}K.
\end{align}

Next, it is easy to verify that
$ \alpha(r)[\alpha(r) - 1][K-L(x_*;r)] = \frac{2}{\sigma^2}(Kr-x_*);$
hence it is an increasing function in $r\geq0$.
Moreover, substituting the explicit representation of $x_*$ in  (\ref{TriggeringBoundaryNon}) into (\ref{InequalityGeneral2}) and (\ref{InequalityGeneral1}) we find that the latter two inequalities are both identical to the following single inequality
\begin{align}\label{ImportantConditionNon}
  \displaystyle\int^{\infty}_0\alpha(r) dF(r)\ge\int^{\infty}_0rdF(r) \displaystyle\int^{\infty}_0\frac{\alpha(r)-1}{r} dF(r).
\end{align}

We have proved the following

\begin{proposition}\label{ExplicitSolutionNon}
The triplet $(V,w,\hat{u})$
solves the Bellman system of the real option problem  if and only if
(\ref{ImportantConditionNon}) holds.
\end{proposition}

Inequality (\ref{ImportantConditionNon}) is a critical condition on the model primitives
we must verify before we can be sure that the solution constructed through the SP is indeed an equilibrium solution to the time-inconsistent real option problem. It is immediate to see that the strict inequality of (\ref{ImportantConditionNon}) is satisfied {\it trivially} when the distribution function $F$ is degenerate corresponding to the classical time-consistent case with an exponential discount function. In this  case, $x_*$ defined by (\ref{TriggeringBoundaryNon}) coincides with the stopping threshold derived in Subsection \ref{tcc}.
This reconciles with the time-consistent setting.

The condition   (\ref{ImportantConditionNon}) may hold for some non-exponential discount functions. Consider a generalized hyperbolic discount function
\begin{equation*}
h(t)=\frac{1}{(1+\gamma t)^{\frac{\beta}{\gamma}}}\equiv \int^{\infty}_0e^{-rt}\frac{r^{\frac{\beta}{\gamma}-1}e^{-\frac{r}{\gamma}}}
{\gamma^{\frac{\beta}{\gamma}}\Gamma(\frac{\beta}{\gamma})}dr,\;\;\gamma>0,\beta>0.
\end{equation*}
We assume that  $\gamma<\beta\le \frac{\sigma^2}{2}.$
Noting that
$
\alpha(r)-1 = -\frac{1}{2}+\frac{\sqrt{\frac{1}{4}\sigma^4+2\sigma^2 r}}{\sigma^2}
$
is a concave function in $r$,  we have
\begin{align*}
\alpha(r)-1 \le (\alpha(r)-1)'\vert_{r=0}r+\alpha(0)-1=\frac{2}{\sigma^2}r.
\end{align*}
Moreover, it is easy to see that
\begin{align*}
\int^{\infty}_0rdF(r) = \beta \text{ and } \alpha(r)\ge 1.
\end{align*}
Therefore,
\begin{align*}
\int^{\infty}_0\frac{\alpha(r)-1}{r}dF(r)\int^{\infty}_0rdF(r)\le \beta \frac{2}{\sigma^2}\le 1\le \int^{\infty}_0\alpha(r)dF(r)
\end{align*}
which is (\ref{ImportantConditionNon}). So, in this case the SP works and the stopping threshold $x_*$ is given by
\begin{align*}
x_*=\displaystyle\frac{\displaystyle\int^{\infty}_0\alpha(r) \frac{r^{\frac{\beta}{\gamma}-1}e^{-\frac{r}{\gamma}}}{\gamma^{\frac{\beta}{\gamma}}\Gamma(\frac{\beta}{\gamma})}dr}{\displaystyle\int^{\infty}_0\frac{\alpha(r)-1}{r} \frac{r^{\frac{\beta}{\gamma}-1}e^{-\frac{r}{\gamma}}}{\gamma^{\frac{\beta}{\gamma}}\Gamma(\frac{\beta}{\gamma})}dr}K.
\end{align*}

However, it is also possible that (\ref{ImportantConditionNon}) fails, which is the case even with the simplest class of non-exponential WDFs -- the pseudo-exponential discount functions. 
To see this, let  $h(t) = \delta e^{-rt}+(1-\delta)e^{-(r+\lambda)t},0<\delta <1, r>0,\lambda>0.$
It is straightforward to obtain that
\begin{align*}
\int^{\infty}_0\alpha(r)dF(r)=\delta \alpha(r)+(1-\delta)\alpha(r+\lambda)
\end{align*}
and
\begin{align*}
\int^{\infty}_0rdF(r) \displaystyle\int^{\infty}_0\frac{\alpha(r)-1}{r} dF(r)>(1-\delta)(r+\lambda)\delta\left(\frac{\alpha(r)-1}{r}\right).
\end{align*}
Since $(1-\delta)(r+\lambda)\delta(\frac{\alpha(r)-1}{r})$ grows faster than $\delta \alpha(r)+(1-\delta)\alpha(r+\lambda)$ when  $\lambda$ becomes large, we conclude that (\ref{ImportantConditionNon}) is violated when $r,\delta$ are fixed and  $\lambda$ is sufficiently large.

What we have discussed so far  shows that the solution constructed through the SP does not solve  the
time-inconsistent real option problem whenever  inequality (\ref{ImportantConditionNon}) fails. A natural question in this case is whether there might
exist equilibrium solutions that cannot be obtained  by the SP  or even by the Bellman system. The answer is resoundingly negative.
\begin{proposition}\label{Prop:Nonexistence}For the real option problem  (\ref{Costfunctional})--(\ref{ExampleNon}), if
 (\ref{ImportantConditionNon}) does not hold,
then no equilibrium stopping rule exists.
\end{proposition}

\begin{proof}
We prove by contradiction. Suppose $\hat{u}$ is an equilibrium stopping rule.
We first note that  $\mathcal{C}_{\hat{u}}\equiv \{x>0: \hat{u}(x)=0\}\neq(0,\infty)$; otherwise
$\hat{u}\equiv 0,$ leading to $J(x;\tau_{\hat{u}})=\int^{\infty}_0L(x;r)dF(r) = \int^{\infty}_0\frac{x}{r}dF(r),$ and hence $J(x;\tau_{\hat{u}})\to \infty$ as $ x\to\infty$  contradicting Lemma \ref{Lemma:PositiveValueFunction} in Appendix A.2.

Define $x_* = \inf\{x:x\in \bar{\mathcal{S}}_{\hat{u}}\}.$ It follows from Lemma \ref{Lemma:Non-existence} in Appendix A.2 that $x_*\in (0,\infty).$ 
A standard argument then leads to
\begin{align*}
J(x;\tau_{\hat{u}})=\int^{\infty}_0\left(K-\frac{x_*}{r}\right)
\left(\frac{x}{x_*}\right)^{\alpha(r)}dF(r)+\int^{\infty}_0\frac{x}{r}dF(r),\;\;x\in (0,x_*].
\end{align*}
Because  $J(x;\tau_{\hat{u}})\le K$ and $J(x_*;\tau_{\hat{u}})=K,$ we have $J_x(x_*-;\tau_{\hat{u}})\ge 0,$ i.e.,
\begin{align*}
\int^{\infty}_0\left(K-\frac{x_*}{r}\right)\alpha(r)\frac{1}{x_*}dF(r)+\int^{\infty}_0\frac{1}{r}dF(r)\ge 0,
\end{align*}
which in turn gives
\begin{align*}
x_*\le \frac{\int^{\infty}_0\alpha(r)dF(r)}{\int^{\infty}_0\frac{\alpha(r)-1}{r}dF(r)}K.\nonumber
\end{align*}
Combining with the failure of  condition (\ref{ImportantConditionNon}), we derive
\begin{align*}
x_* <\int^{\infty}_0rdF(r),
\end{align*}
which contradicts Lemma \ref{Lemma:Non-existence}. This completes the proof.
\end{proof}

The above is a stronger result. It suggests that for the problem to have  any equilibrium stopping rule at all (not necessarily the one obtainable by the SP principle), condition (\ref{ImportantConditionNon}) {\it must} hold. So, when it comes to a time-inconsistent stopping problem with non-exponential discounting, it is highly likely that no equilibrium stopping rule exists, even if the SP principle does generate a ``solution", or even if the time-consistent counterpart (in which everything else is identical except the discount function) is indeed solvable  by the SP. Applying these conclusions to the pseudo-exponential discount functions discussed above, we deduce that there is no equilibrium stopping when $\lambda$ is sufficiently large.

Having said this, a logical conclusion from Propositions \ref{ExplicitSolutionNon} and \ref{Prop:Nonexistence} is that when
equilibria do exist,  one  of them must be a solution generated by the SP.
In general if there exists an equilibrium then there may be multiple ones;
see,  for example, \cite{krusmi2003} and \cite{ekepir2008} for multiple equilibria in time-inconsistent control problems. In this case, the SP can only generate {\it one} of them, but not necessarily {\it all} of them. (This statement is true even for a  classical
time-consistent stopping problem.) So, after all, the SP   is still a useful, proper method for time-inconsistent
problems; we can simply apply it to generate a candidate solution. If the solution is an equilibrium (which we must verify), then we have found one (but not necessarily other equilibria); if
it is not an equilibrium, then we know there is no equilibrium at all.


\section{Conclusions}\label{Conclusion}
While the SP principle has been widely used to study time-inconsistent stopping problems, our results indicate the risk of using this principle on such problems. We have shown that the SP principle solves the time-inconsistent problem if and only if certain inequalities are satisfied.  \par
By a simple model of the classical real option problem, we have found that these inequalities may be violated even for  simple and commonly used non-exponential discount functions. When the SP principle fails, we have shown the intra-personal equilibrium does not exist. The nonexistence result and the failure of the SP principle suggest that it is imperative that the techniques for conventional optimal stopping problems be used more carefully when extended to solving time-inconsistent stopping problems.
\appendix
\renewcommand{\thesubsection}{\Alph{section}.\arabic{subsection}}
\section{Appendix: Proofs}
%

\subsection{Proof of Theorem \ref{EquilibriumcharacterizationNon}}
For the stopping time $\tau^{\epsilon,a}$, if $a = 1$, then $J( x; \tau^{\epsilon,a}) = g(x)$. The Bellman equation (\ref{StationaryPDENon}) implies that $g(x)\ge V(x)\equiv J(x;\tau_{\hat u})$.
This yields (\ref{StationaryDefinitionInequalityNon}).\par
If $a=0$, then
\begin{align}\label{Equ:P}
J(x;\tau^{\epsilon,a})&=\mathbb{E}\left[\int^{\epsilon}_0h(s)f(X_s)ds \Big\vert X_0=x\right]+\mathbb{E}\left[\int^{\tau^{\epsilon,a}}_{\epsilon}(h(s)-h(s-\epsilon))f(X_s)ds\Big\vert X_0=x\right]\nonumber\\&+\mathbb{E}[(h(\tau^{\epsilon,a})-h(\tau^{\epsilon,a}-\epsilon))
g(X_{\tau^{\epsilon,a}})\vert X_t=x]+\mathbb{E}[V(X_{\epsilon})\vert X_0=x]\nonumber\\&=\mathbb{E}\left[\int^{\epsilon}_0h(s)f(X_s)ds\Big\vert X_0=x \right]+\mathbb{E}\left[\int^{\tau^{\epsilon,a}}_{\epsilon}\int^{\infty}_0
e^{-r(s-\epsilon)}(e^{-\epsilon r}-1)dF(r)f(X_s)ds\Big\vert X_0=x\right]\nonumber\\&+\mathbb{E}\left[\int^{\infty}_0e^{-r(\tau^{\epsilon,a}-\epsilon)}
(e^{-\epsilon r}-1)dF(r)g(X_{\tau^{\epsilon,a}})\Big\vert X_0=x\right]+\mathbb{E}[V(X_{\epsilon})\vert X_0=x]\nonumber\\&=\mathbb{E}\left[\int^{\epsilon}_0h(s)f(X_s)ds\Big\vert X_0=x\right]+\mathbb{E}\left[\int^{\infty}_0(e^{-\epsilon r}-1)w(X_{\epsilon};r)dF(r)\Big\vert X_0=x\right]\nonumber\\&+\mathbb{E}[V(X_{\epsilon})\vert X_0=x].
\end{align}
Define $\tau_n = \inf\{s\ge 0: \sigma(X_s)V_x(X_s) > n   \}\wedge\epsilon$. Then it follows from Ito's formula that
\begin{align*}
\mathbb{E}\left[ V(X_{\tau_n})\vert X_0=x\right]= \mathbb{E}\left[\int^{\tau_n}_0(\frac{1}{2}\sigma^2(X_s)V_{xx}(X_s)+b(X_s)V_x(X_s))ds\Big\vert X_0=x\right]+V(x).
\end{align*}
By (\ref{StationaryPDENon}), we conclude
\begin{align*}
\mathbb{E}\left[ V(X_{\tau_n})\vert X_0=x\right]&= \mathbb{E}\left[\int^{\tau_n}_0(\frac{1}{2}\sigma^2(X_s)V_{xx}(X_s)+b(X_s)V_x(X_s))ds\Big\vert X_0=x\right]+V(x)\nonumber\\&\ge
\mathbb{E}\left[\int^{\tau_n}_0(-f(X_s)+\int^{\infty}_0rw(X_s;r)dF(r))ds\Big\vert X_0=x\right]+V(x).
\end{align*}
Note that conditions (\ref{GrowthPolynomial}) and (\ref{GrowthPolynoimalCostFunctional}) ensure that $-f(x)+\int^{\infty}_0rw(x;r)dF(r)$ has polynomial growth, i.e., there exist $C>0,m\ge 1$ such that
\begin{align*}
\Big\vert -f(x)+\int^{\infty}_0rw(x;r)dF(r)\Big\vert \le C(\vert x\vert^m+1),
\end{align*}
which leads to
\begin{align*}
\sup_{0\le t\le \epsilon}\Big\vert-f(X_s)+\int^{\infty}_0rw(X_s;r)dF(r)\Big\vert\le C(\sup_{0\le t\le \epsilon}\vert X_t \vert^m+1).
\end{align*}
Moreover, under condition (\ref{Lip}), it follows from standard SDE theory (see, for example, Chapter $1$ of \cite{yong1999stochastic}) that equation (\ref{SDE}) admits a unique strong solution $X$ satisfying
\begin{align*}
\mathbb{E}[\sup_{0\le t\le \epsilon}\vert X_t\vert^m\vert X_0=x]\le K_{\epsilon}(\vert x \vert^m + 1)
\end{align*}
with $K_{\epsilon}>0.$\par
Then letting $n\to \infty$, we conclude by the dominated convergence theorem that
\begin{align*}
\mathbb{E}\left[ V(X_{\epsilon})\vert X_0=x\right]&\ge
\mathbb{E}\left[\int^{\epsilon}_0(-f(X_s)+\int^{\infty}_0rw(X_s;r)dF(r))ds\Big\vert X_0=x\right]+V(x).
\end{align*}
Consequently,
\begin{align*}
&\liminf_{\epsilon\rightarrow 0+}\frac{J(x;\tau^{\epsilon,a})-J(x;\tau_{\hat{u}})}{\epsilon}\\&\ge\liminf_{\epsilon\rightarrow 0+} \mathbb{E}\left[\int^{\epsilon}_0h(s)f(X_s)ds\Big\vert X_0=x\right]+\mathbb{E}\left[\int^{\infty}_0(e^{-\epsilon r}-1)w(X_{\epsilon};r)dF(r)\Big\vert X_0=x\right]\nonumber\\&+\liminf_{\epsilon\rightarrow 0+}\frac{1}{\epsilon}\mathbb{E}\left[\int^{\epsilon}_0\int^{\infty}_0(rw(X_t;r)dF(r)-f(X_t))dt\Big\vert X_0=x \right].
\end{align*}
The continuity of $f$ and $w$ along with the  polynomial growth conditions (\ref{GrowthPolynomial}) and (\ref{GrowthPolynoimalCostFunctional}) allow the use of the dominated convergence theorem, which yields
\begin{align*}
\liminf_{\epsilon\rightarrow 0+}\frac{J(x;\tau^{\epsilon,a})-J(x;\tau_{\hat{u}})}{\epsilon}&\ge0.
\end{align*}
This completes the proof.
\subsection{Proof of Proposition \ref{Prop:Bechmark}}
Let $V^B$ be  the value function of the optimal stopping problem. It follows from the standard argument (see, for example, Chapter $6$ of \citeauthor{krylov2008controlled} \citeyear{krylov2008controlled}) that $V^B$ is continuously differentiable and its first-order derivative is absolutely continuous. Moreover, $V^B$ solves the following Bellman equation
\begin{equation}\label{BellmanEquationConvention}
\min\left\{\frac{1}{2}\sigma^2x^2V^B_{xx}(x)+bxV^B_x(x)+f(x)-rV^B(x),K-V^B(x)\right\}=0.
\end{equation}
Define the continuation region  $\mathcal{C}^B = \{x>0:V^B(x)<K\}$ and the stopping region $\mathcal{S}^B = \{x>0:V^B(x) = K\}$.

We claim that $\mathcal{S}^B \neq (0,\infty).$ If not, then $V^B\equiv K.$ Thus $\frac{1}{2}\sigma^2x^2V^B_{xx}(x)+bxV^B_x(x)+f(x)-rV^B(x)<0$ whenever $x\in\{x>0:f(x)-rK<0\}.$ However, since $f(0)<rK,$ the continuity of $f$ implies  $\{x>0:f(x)-rK<0\}\neq\emptyset.$ This contradicts the Bellman equation (\ref{BellmanEquationConvention}).\par
We now show that $\mathcal{C}^b \neq (0,\infty).$ If it is false, then  we have $V^B(x) = L(x;r)$, with $L$ defined by (\ref{SpecialSolution}). Since $f$ is increasing and bounded from below by $0$, we have
$$V^B(\infty)\equiv  \lim_{x\to\infty}V^B(x)= \int^{\infty}_0\int^{\infty}_0\lim_{x\to\infty}f(yx)e^{-rs}G(y,s)dyds.$$
The concavity of $f$ yields  $f(x)\ge xf_x(x)+f(0)$. It then follows from $\lim_{x\to\infty}xf_x(x) = \infty$ that $\lim_{x\to\infty}f(x) = \infty$, which yields that $V^B(\infty)=\infty$. This contradicts the fact that $V^B(x)\le K.$\par
Next, since $X$ is a geometric Brownian motion and $f$ is increasing, it is clear that $V$ is increasing too.
Now, we derive the value of the triggering boundary, $x_B$,  via the SP principle. Specifically, it follows from  (\ref{BellmanEquationConvention}) that
\begin{align*}
V^b(x) &= (K-L(x_B;r))(\frac{x}{x_B})^{\alpha(r)} + L(x;r),\;\;x<x_B\\
V^b(x) &= K, x\ge x_B,
\end{align*}
where $\alpha(r)$ is defined by (\ref{Root}).
Then the SP implies that $V^{B}_x(x_B)=0$ which after some calculations
yields that $x_B$ is the solution of the equation (\ref{SPEquation}).\par
To prove the unique existence  of the solution of (\ref{SPEquation}), define $Q(x) := \alpha(r)(K-L(x;r)) + L_x(x;r)x.$ Then $Q_x(x) = (-\alpha(r) +1)L_x(x;r) + L_{xx}(x;r)x$. As $L$ is  strictly increasing and concave and $\alpha(r)>1$, we deduce that $Q$ is strictly decreasing.  It remains  to show $Q(0)>0$ and $Q(\infty)<0$. It is easy to see that $Q(0) =\alpha(r)(K-L(0;r)) = \alpha(r)(K-\frac{f(0)}{r}) >0$ and $Q(x) = \alpha(r)(K-L(0;r)-\int^{x}_0L_x(s;r)ds)+L_x(x;r)x.$ Since $L$ is concave,  we have $\int^x_0L_x(s;r)ds\ge xL_x(x;r).$ Thus $Q(x)\le\alpha(r)(K-L(0;r)) +(-\alpha(r)+1)xL_x(x;r)$. Recalling that $\lim_{x\to\infty}xL_x(x;r) = \infty$ and $\alpha(r)>1$, we have $Q(\infty) = -\infty$. This completes the proof.

\subsection{Proof of Theorem \ref{Prop:SPGeneral}}
We need to  present a series of lemmas before giving a proof of Theorem \ref{Prop:SPGeneral}.
\begin{lemma}\label{Lem:Regularity}
Given a stopping rule $u$ and a discount rate $r>0,$ the function $E(x;\tau_u,r):=\mathbb{E}[\int^{\tau_u}_0e^{-rt}f(X_t)dt+e^{-r\tau_u}K\vert
X_0=x]$ is  continuous in $x\in(0,\infty)$.
\end{lemma}
\begin{proof}
We prove the right continuity of $E(\cdot;\tau_u,r)$ at a given $x_0>0$; the left  continuity can be discussed in the same way.\par

If there exists $\delta >0$ such that $(x_0,x_0+\delta)\in \mathcal{S}_u$, then the right continuity of $E(\cdot;\tau_u,r)$ at $x_0$ is obtained immediately.
If there  exists $\delta >0$ such that $(x_0,x_0+\delta)\in \mathcal{C}_u$, then it follows from the Feynman-Kac formula that $E(\cdot;\tau_u,r)$ is the solution to the differential equation $\frac{1}{2}\sigma^2x^2E_{xx}+bxE_x-rE+f=0$ on $(x_0,x_0+\delta)$. This in particular implies that $E(\cdot;\tau_u,r)\in C^2((x_0,x_0+\delta))\cap C([x_0,x_0+\delta])$ due to the regularity of $f$ and the coefficients of the differential equations; hence the right continuity of $E(\cdot;\tau_u,r)$ at $x_0$.

Otherwise, we first assume that $f(x_0)\ge rK$ and consider the set $\mathcal{C}_u\cap(x_0,\infty).$ Since it is an open set, we have $\mathcal{C}_u\cap(x_0,\infty)=\cup_{n\ge 1}(a_n,b_n),$ where $a_n,b_n \in \bar{\mathcal{S}_u}, \forall n\ge1$. It is then easy to see that $x_0$ is an accumulation point of $\{a_n\}_{n\ge 1}$ and hence $x_0\in\bar{\mathcal{S}_u}$.\par
Define $I(x):=E(x;\tau_u,r)-K$ for $x\in (a_n,b_n).$ It is easy to see that $I$ solves the following differential equation \begin{align}\label{AuxiliaryFunction}
\frac{1}{2}\sigma^2x^2I_{xx}(x)+bxI_x(x)-rI(x)+f(x)-rK=0.
\end{align}
with the boundary conditions
\begin{align*}
I(a_n)=I(b_n) = 0.
\end{align*}
Consider an auxiliary function $H$ that solves the following differential equation
\begin{align*}
\frac{1}{2}\sigma^2x^2H_{xx}(x)+bxH_x(x)-rH(x)+f(x)-rK=0,
\end{align*}
with the boundary conditions
\begin{align*}
H(x_0)=H(b_1) = 0.
\end{align*}
Since $f(x)>rK$ on $(x_0,\infty),$ the comparison principle shows that $ H(x)\ge 0, \forall x\in [x_0,b_1].$ Applying the comparison principle again on any $(a_n,b_n)\cap(x_0,b_1), \forall n\in\mathbb{N}^+,$ we have  $0\le I(x)\le H(x).$ Noting that  $H(x)\rightarrow H(x_0)=0$ as $x\rightarrow x_0+$, we conclude that $I(\cdot)$ is right continuous at $x_0$ and so is  $E(\cdot;\tau_u,r)$.\par
For the case $f(x_0)< rK,$ a similar argument applies. Indeed, consider an auxiliary function $H_1$  satisfying the differential equation (\ref{AuxiliaryFunction}) on $(x_0,f^{-1}(rK))$ with the boundary condition $H_1(x_0)=H_1(f^{-1}(rK))=0.$ The comparison principle yields that $H_1(x)\le I(x)\le 0$ on $(a_n,b_n)\cap(x_0,f^{-1}(rK)), \forall n\in\mathbb{N}^+$. The right continuity of $I(\cdot)$ and $E(\cdot;\tau_u,r)$ then follows immediately.
\end{proof}
\begin{lemma}\label{Lemma:PositiveValueFunction}
If $\hat{u}$ is an equilibrium stopping rule, then $J(x;\tau_{\hat{u}})\le K\; \forall x\in(0,\infty).$
\end{lemma}
\begin{proof}
If there exists $x_0\in(0,\infty)$ such that $J(x_0;\tau_{\hat{u}})>K,$ then we have
\begin{align*}
\limsup_{\epsilon\rightarrow 0}\frac{J(x_0;\tau_{\hat{u}})-J(x_0;\tau^{\epsilon,1})}{\epsilon}=\infty,
\end{align*}
where $\hat{u}^{\epsilon,1}$ is given by (\ref{uAuxiliaryStationaryNon}).
This contradicts the definition of an equilibrium stopping rule.
\end{proof}
\begin{lemma}\label{Lemma:Non-existence} If $\hat{u}$ is an equilibrium stopping rule, then $\{x>0: f(x)<\int^{\infty}_0rdF(r)K\}\subset \mathcal{C}_{\hat{u}}.$
\end{lemma}
\begin{proof}
Suppose that there exists $x\in \{x>0: f(x)<\int^{\infty}_0rdF(r)K\}\cap\bar{\mathcal{S}}_{\hat{u}},$ then it follows from Lemma \ref{Lemma:PositiveValueFunction} that $\mathbb{E}[J(X_t;\tau_{\hat{u}})\vert X_0=x]\le K.$ Consider the stopping time $\tau^{\epsilon,0}.$ Equation  (\ref{Equ:P}) and the fact that $J(x;\tau_{\hat{u}})=K$ give
\begin{align*}
\frac{J(x;\tau^{\epsilon,0})-J(x;\tau_{\hat{u}})}{\epsilon}&\le\frac{1}{\epsilon}\mathbb{E}\left[\int^{\epsilon}_0h(s)f(X_s) ds \Big\vert X_0=x\right]+\mathbb{E}\left[\int^{\infty}_0\left(\frac{e^{-\epsilon r}-1}{\epsilon}\right)w(X_{\epsilon};r)dF(r)\Big\vert X_0=x\right].
\end{align*}
As $w(\cdot,r)$ is continuous (Lemma \ref{Lem:Regularity}) and $w(x;r)=K$, we have
\begin{align*}
\liminf_{\epsilon\rightarrow 0}\frac{J(x;\tau^{\epsilon,0})-J(x;\tau_{\hat{u}})}{\epsilon}&\le  f(x)-\int^{\infty}_0rKdF(r)<0.
\end{align*}
This contradicts the definition of an equilibrium stopping rule.
\end{proof}\par
We now turn to the proof of Theorem \ref{Prop:SPGeneral}. We begin  with the sufficiency. To this end it suffices to show that $V(x)\le K, x\in(0,x_*)$ and $f(x)-\int^{\infty}_0rdF(r)K\ge 0, x\in(x_*,\infty).$ \par
We first show that $V_{xx}\le 0, x\in(0,x_*).$ By simple algebra, we have
\begin{align*}
V_{xx}(x) = \int^{\infty}_0\alpha(r)(\alpha(r) - 1)(K-L(x;r))(\frac{x}{x_*})^{\alpha(r)}\frac{1}{x^2}dF(r) + \int^{\infty}_0L_{xx}(x;r)dF(r).
\end{align*}
As $L$ is concave, we only need to prove $\int^{\infty}_0\alpha(r)(\alpha(r) - 1)(K-L(x;r))(\frac{x}{x_*})^{\alpha(r)}dF(r)\le 0.$ It is easy to see that $(\frac{x}{x_*})^{\alpha(r)}$ is decreasing in $r$ given that $\alpha(r)$ is increasing in $r$ and $x<x_*.$ Then the rearrangement inequality (e.g., Chapter $10$ of \citeauthor{hardy1952inequalities} \citeyear{hardy1952inequalities}; \citeauthor{lehmann1966some} \citeyear{lehmann1966some}) yields\footnote{Inequality (\ref{InequalityRearrangement}) can be read as \begin{align*}
\text{cov}(X,Y) \le 0,
\end{align*} with $X = \alpha(R)(\alpha(R) - 1)(K-L(x;R))(\frac{x}{x_*})^{\alpha(R)}$ and $Y=(\frac{x}{x_*})^{\alpha(R)}$, where $R$ is a random variable with distribution function $F$. Because of the monotonicity of $X,Y$ in $R$, $X$ and $Y$ are anti-comonotonic. Then inequality (\ref{InequalityRearrangement}) follows from the fact that the covirance of two anti-comonotonic random variables is non-positive. }\begin{align}\label{InequalityRearrangement}
&\int^{\infty}_0\alpha(r)(\alpha(r) - 1)(K-L(x;r))(\frac{x}{x_*})^{\alpha(r)}dF(r)\nonumber\\&\le \int^{\infty}_0\alpha(r)(\alpha(r) - 1)(K-L(x;r))dF(r)\int^{\infty}_0(\frac{x}{x_*})^{\alpha(r)}dF(r).
\end{align}
Therefore it follows from (\ref{InequalityGeneral1}) that $V_{xx}(x)\le0,\; x\in(0,x_*).$ Now, $V_x(x_*) = 0.$  Thus $V_x(x) \ge 0$ and consequently $V(x)\le K$
$\forall x\in (0,x_*)$, due to  $V(x_*) = K.$\par
Next, the inequality $f(x)-\int^{\infty}_0rdF(r)K\;\forall  x\in(x_*,\infty)$ follows from $f$ being increasing along with inequality (\ref{InequalityGeneral2}). This completes the proof of the sufficiency.\par
We now turn to the necessity part. As (\ref{InequalityGeneral2}) is an immediate corollary of Lemma \ref{Lemma:Non-existence}, we only need to prove (\ref{InequalityGeneral1}). Suppose (\ref{InequalityGeneral1}) does not hold. Then by a simple calculation, we have  $$V_{xx}(x_*-) =  \int^{\infty}_0\alpha(r)(\alpha(r) - 1)(K-L(x_*;r))\frac{1}{x_*^2}dF(r)+\int^{\infty}_0L_{xx}(x_*;r)dF(r)>0.$$ However, $V_x(x_*) = 0$,  implying  that there exists  $x_1\in(0,x_*)$ such that $V_x(x)<0$ on $x\in(x_1,x_*).$ Then it follows from $V(x_*) =K$ that $V(x)>K$ when $x\in(x_1,x_*),$ which contradicts Lemma \ref{Lemma:PositiveValueFunction}.

\bibliography{EbertReferences}        
\bibliographystyle{ecta}  

\end{document}